\documentclass{article}
\usepackage[utf8]{inputenc}

\usepackage[T1]{fontenc}
\usepackage{soul, comment} 
\usepackage{phfqit}
\usepackage{proba} 
\usepackage[sort,nocompress]{cite}
\usepackage{amssymb}
\usepackage{amsmath}
\usepackage{amsthm}
\usepackage{amsfonts}
\usepackage{mathtools}
\usepackage{xspace}
\usepackage{bm}
\usepackage{nicefrac}
\usepackage{commath}
\usepackage{bbm}
\usepackage{boxedminipage}
\usepackage{xparse}
\usepackage{xcolor}
\usepackage{float}
\usepackage{multirow}
\usepackage{makecell}
\usepackage{graphicx}
\usepackage{caption}
\usepackage{subcaption}
\usepackage{ifthen}
\usepackage{thm-restate}
\usepackage{algorithm}
\usepackage[noend]{algpseudocode}
\usepackage{colortbl} 
\usepackage{fancyhdr}   
\usepackage{hyperref}
    \hypersetup{ colorlinks=true, linkcolor=blue, citecolor=magenta, urlcolor=blue,}
\usepackage{fullpage}
\usepackage[utf8]{inputenc}
\usepackage{environ}
\usepackage{mdframed}
\usepackage{cleveref}
\usepackage[short]{optidef} 
\usepackage{enumitem}
\usepackage{tikz}
\usetikzlibrary{shapes}
\usetikzlibrary{positioning}
\usetikzlibrary{fit}
\usetikzlibrary{graphs,graphs.standard}
\usepackage{enumitem}

\newtheorem{lemma}{Lemma}[section]
\newtheorem{remark}{Remark}[section]
\newtheorem{theorem}{Theorem}[section]

\newtheorem{definition}{Definition}[section]

\newtheorem{claim}{Claim}[section]

\usepackage{enumitem}

\NewEnviron{problem}[1]{%
	\begin{center}\fbox{\parbox{6in}{%
				{\centering\scshape #1\par}%
				\parskip=1ex
				\everypar{\hangindent=1em}%
				\BODY
}}\end{center}}

\title{Hypergraph Splitting-Off \\ via Element-Connectivity Preserving Reductions\thanks{Grainger College of Engineering, University of Illinois, Urbana-Champaign. Email: \{karthe, chekuri, smkulka2\}@illinois.edu. Supported in part by NSF grant CCF-2402667.}}

\author{Karthekeyan Chandrasekaran \and Chandra Chekuri \and Shubhang Kulkarni}
\date{}
\begin{document}
\maketitle
\begin{abstract}
  Bérczi, Chandrasekaran, Király, and Kulkarni \cite{Bérczi_Chandrasekaran_Király_Kulkarni_2023} recently described a splitting-off procedure in hypergraphs that preserves local-connectivity and outlined some applications.
In this note we give an alternative proof via element-connectivity preserving reduction operations in graphs \cite{HindO96,Chekuri_Korula_2009,CK14}.  
\end{abstract}


\section{Introduction}
\label{sec:intro}
\paragraph{Graphs and Hypergraphs.} A \emph{hypergraph} is a pair $H = (V, E)$, where $V$ is a set of vertices and $E \subseteq 2^V \setminus \{\emptyset\}$ is a multiset of hyperedges (each a subset of $V$). A graph is a hypergraph in which each hyperedge has cardinality at most two. We will assume that both graphs and hypergraphs in this paper do not have edges of cardinality one (i.e., no self-loops).  Given a hypergraph $H=(V,E)$ and a set $S \subseteq V$ we let $\delta_H(S) = \{e \in E \mid e \cap S \neq \emptyset, e \cap (V-S) \neq \emptyset\}$ denote the set of hyperedges crossing $S$, that is, those hyperedges with at least one vertex in $S$ and at least one vertex in $V-S$. Given two distinct vertices $s,t \in V$, we let $\lambda_H(s,t)$ denote the minimum-cut between $s$ and $t$ in $H$; formally $\lambda_H(s,t) = \min_{S: S \cap \{s,t\} =1} |\delta_H(S)|$. 
The incidence graph of a hypergraph $H=(V,E)$ is a bipartite graph $G_H=(U_V \uplus U_E, F)$ where 
one side has a vertex-node $u_v$ for each $v \in V$ and the other side has a hyperedge-node $u_e$ for each $e \in E$, and there is an edge between $u_v$ and $u_e$ iff $v \in e$.  It is well-known that $\lambda_H(s,t)$ can be computed as follows: we compute the max $s$-$t$ flow in $G_H$ where each hyperedge-node $(u_e)_{e \in E}$ is given unit \emph{vertex} capacity and each $u_v, v \in V-\{s,t\}$ is given infinite vertex capacity (we can also replace infinity with the degree of the vertex).  For graphs, we can simplify this by simply computing the max $s$-$t$ flow in the original graph with unit edge capacities.

\paragraph{Splitting-off in graphs.} 
Splitting off in graphs is an operation that was introduced by Lovász \cite{Lovasz} and has found many applications. The operation is to replace two edges $su$ and $sv$ that are incident to a vertex $s$ by the edge $uv$ that bypasses $s$. The goal is to do this operation while preserving the edge-connectivity between vertices that do not involve $s$. In many applications, this operation is repeated at a vertex $s$ until $s$ is isolated and hence can be removed. The initial work of Lovász was in the context of preserving the global edge-connectivity of the graph, and in particular, he considered Eulerian graphs. Subsequently, Mader proved a stronger theorem on splitting-off to preserve local-edge-connectivity. 

\begin{theorem}[Mader \cite{Mader}]
  \label{thm:mader}
  Let $G=(V\cup \{s\},E)$ be an undirected multi-graph, where $deg(s)
  \neq 3$ and $s$ is not incident to a cut edge of $G$.  Then $s$
  has two neighbours $u$ and $v$ such that the graph $G'$ obtained
  from $G$ by replacing $su$ and $sv$ by $uv$ satisfies
  $\lambda_{G'}(x,y) = \lambda_G(x,y)$ for all $x,y \in V \setminus
  \{s\}$.
\end{theorem}

The splitting-off operation has been considered for directed graphs as well, however, they are limited to Eulerian digraphs and graphs that satisfy similar requirements --- see \cite{Mader,Frank88,BFJ95}. In this work, we will focus on undirected hypergraphs and graphs. 

While splitting-off has many applications for edge-connectivity problems, it is not quite applicable for problems related to vertex-connectivity in graphs. In the same vein, there is no simple analogue of splitting-off in hypergraphs. In graphs, a notion called element-connectivity has played an important role in helping with vertex-connectivity problems. Element-connectivity admits a graph reduction step, and this has found several applications. This reduction step was originally introduced in the context of global-connectivity by Hind and Oellerman \cite{HindO96}, and was subsequently extended to the local-connectivity setting by Chekuri and Korula \cite{Chekuri_Korula_2009,CK14}. One of the consequences of the element-connectivity reduction is that it leads to a hypergraph --- we will describe the formal claim shortly. Very recently Bérczi, Chandrasekaran, Király, and Kulkarni \cite{Bérczi_Chandrasekaran_Király_Kulkarni_2023} described a splitting-off procedure in hypergraphs that preserves local-connectivity and showed some interesting applications.  The purpose of this work is to point out that the hypergraph splitting-off result can be derived via the element-connectivity preserving reduction operation.  The proof is short, and we believe that it also provides an alternative understanding of the hypergraph splitting-off operation.

We note that \cite{Bérczi_Chandrasekaran_Király_Kulkarni_2023} provides a strongly polynomial-time algorithm for the hypergraph splitting-off operation in the capacitated case, which is non-trivial to prove, and an analogue of such a result is not known in the element-connectivity setting so far. The rest of the paper is organized as follows. We first define element-connectivity and the reduction operation theorem from \cite{CK14}. We then describe the hypergraph splitting-off operations and the result from \cite{Bérczi_Chandrasekaran_Király_Kulkarni_2023}. We relate the two in the subsequent section and derive the splitting-off theorem in hypergraphs via element-connectivity.

\subsection{Element-Connectivity Preserving Graph Reduction}\label{sec:introduction}

Let $G =(V, E)$ be an undirected graph and $T\subseteq V$ be a set of \emph{terminal} vertices. The vertices in the set $V - T$ are referred to as \emph{non-terminals}. The non-terminals and edges of the graph are collectively called \emph{elements}. For two distinct terminals $u, v \in T$, the $(u,v)$-\emph{element-connectivity}, denoted by $\kappa'_{(G, T)}(u,v)$ is defined as follows:
$$\kappa'_{(G, T)}(u,v) := \min\left\{|F|: F \subseteq (V - T) \cup E \text{ where $u$ and $v$ are disconnected in $G - F$}\right\},$$
i.e., the minimum number of elements required to be deleted in order to disconnect the terminals $u$ and $v$.

\begin{figure}[htb]
    \centering
    \includegraphics[width=0.9\linewidth]{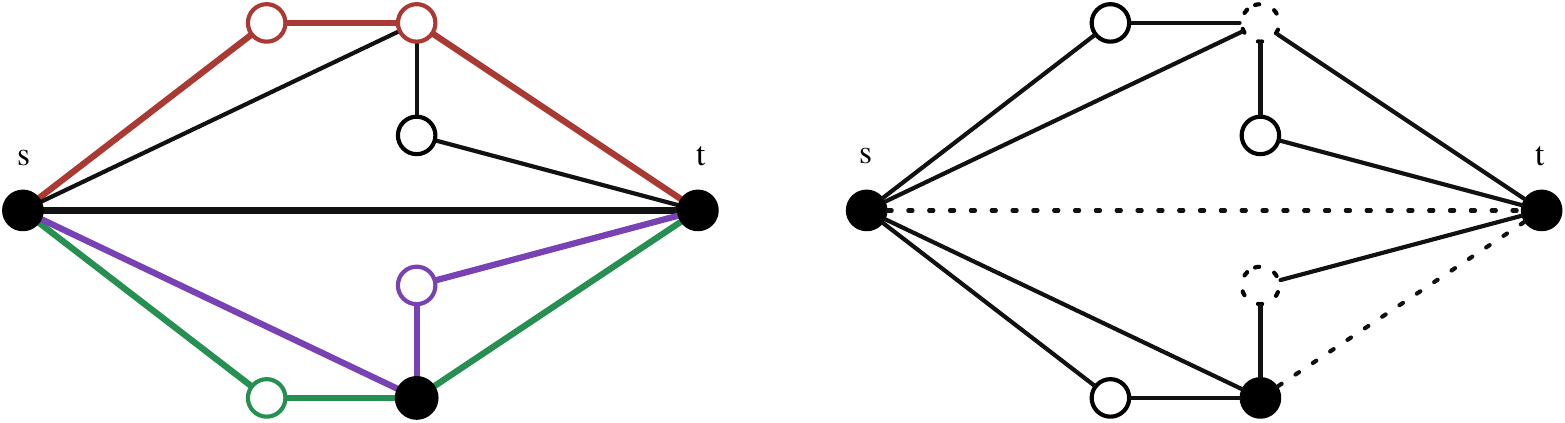}
    \caption{Example of element connectivity from \cite{CRX15}. The terminals are shown in filled black circles. $s$ and $t$ are $3$ element connected; 3 element-disjoint paths and a 3 element-cut are shown.}
    \label{fig:elem-conn}
\end{figure}

The following lemma is easy to deduce from Menger's theorem (equivalently, maxflow-mincut) and is well-known in the literature, so we do not provide
a proof.

\begin{lemma}
  \label{lem:elem-basic}
  For every pair of distinct terminals $u,v\in T$, the element-connectivity $\kappa'_{(G, T)}(u,v)$ is 
  equal to the maximum number of pairwise element-disjoint $(u, v)$ paths in $G$. Moreover, $\kappa'_{(G, T)}(u,v)$ can be computed efficiently by solving a max $u$-$v$-flow problem where each edge has unit edge capacity, each non-terminal has unit vertex capacity, and each
  terminal that is not $u$ or $v$ has vertex capacity equal to its degree.
\end{lemma}

For $e \in E$, we let $G_{/e}$ denote the graph obtained from $G$ by contracting the edge $e$ and $G - e$ denote the graph with edge $e$ deleted. The next theorem by Chekuri and Korula states that for every edge between non-terminal vertices, at least one of the graphs obtained by either deleting the edge or contracting the edge maintains all pairwise element connectivities amongst terminals. 
\begin{theorem}\label{thm:elem-conn-preserving-reduction}
\cite{Chekuri_Korula_2009}
    Let $G =(V, E)$ be an undirected graph and $T \subseteq V$ be a set of terminals. Let $pq$ be an arbitrary edge where $p, q \in V - T$ and let $G_1 := G - pq$ and $G_2 := G_{/pq}$. Then, at least one of the following holds: 
    \begin{enumerate}
        \item $\kappa'_{(G_1, T)}(u,v) = \kappa'_{(G,T)}(u,v)$ $\forall u, v \in T$
        \item $\kappa'_{(G_2, T)}(u,v) = \kappa'_{(G,T)}(u,v)$ $\forall u, v \in T$
    \end{enumerate}
\end{theorem}

We will denote the operation of deleting or contracting an edge to preserve element-connectivities between all terminal pairs as an element-connectivity preserving reduction operation. 

\begin{remark}
  By repeatedly applying the above theorem we transform the original graph $G=(T \cup (V-T),E)$ into a new graph $H$ such that (i) $H$ is a minor of $G$ and (ii) the non-terminals form a stable set in $H$. We can also assume that the terminals form a stable set by sub-dividing each edge between two terminals and placing a non-terminal on it. This resulting bipartite graph can be viewed as a hypergraph $H'=(T,E_H)$ with vertex set $T$, and for any $u,v \in T$, $\kappa'_H(u,v)$ is the same as the hypergraph edge-connectivity $\lambda_{H'}(u,v)$.  
\end{remark}

We state and prove the following useful lemma that is intuitive but has not been explicitly noted so far.

\begin{lemma}\label{lem:deletions-later-allow-deletions-now}
Let $G=(V,E)$ be an instance of an element-connectivity instance with $T$ as the terminal set. Let $G'=(V',E')$ be a minor of $G$ after a sequence of element-connectivity preserving operations. Suppose deleting an edge $p'q' \in E'$ between two non-terminals is element-connectivity preserving. Let $pq$ be the edge in $G$ corresponding to $p'q'$. Then deleting $pq$ in $G$ is an element-connectivity preserving operation. 
\end{lemma}
\begin{proof}
Consider two distinct terminals $s,t$. Since $p'q'$ can be deleted in $G'$, there are $\kappa'_{G'}(s,t) = \kappa'_G(s,t)$ element-disjoint $s$-$t$ paths in $G'-p'q'$. Let $\calP'$ be such a collection of paths. Since $G'-p'q'$ is a minor of $G$ where only edges between non-terminals are deleted or contracted, each vertex in $G'$ corresponds to a contraction of a connected subgraph consisting only of non-terminals of $G$.
Thus the paths in $\calP'$ can be mapped to a collection $\calP$ of element disjoint $s$-$t$ paths in $G-pq$ such that $|\calP'| = |\calP|$. Thus $\kappa'_{G-pq}(s,t) = \kappa'_G(s,t)$. Since $s,t$ were arbitrary, we see that deleting $pq$ in $G$ is an element-connectivity preserving operation.
\end{proof}

The preceding lemma shows that if an edge $pq$ between non-terminals is not deletable at some stage then 
it does not become eligible for deletion after a sequence of other element-connectivity preserving reduction operations and thus, $p$ and $q$ will be merged eventually.

\subsection{Connectivity Preserving Hypergraph Splitting-Off}
We begin with a basic fact relating hyperedge-connectivity and hyperedge-disjoint paths in a hypergraph. A $u-v$ path in a hypergraph $H=(V, E)$ is a sequence $u_1=u, e_1, u_2, e_2, \ldots, u_{k-1}, e_{k-1}, u_k=v$, where $u_i, u_{i+1}\in e_i$ for every $i\in [k-1]$. The following lemma is a consequence of Menger's theorem. 
\begin{lemma}
  \label{lem:hyperedge-connectivity}
  Let $H=(V, E)$ be a hypergraph. 
  For every pair of distinct vertices $u,v\in V$, the hyperedge-connectivity $\lambda_{H}(u,v)$ is 
  equal to the maximum number of pairwise hyperedge-disjoint paths between $u$ and $v$ in $G$. 
\end{lemma}

Bérczi, Chandrasekaran, Király, and Kulkarni defined the hypergraph splitting-off operation as follows. 

\begin{definition}\label{defn:splitting-off}
\cite{Bérczi_Chandrasekaran_Király_Kulkarni_2023}
    Let $H=(V+s, E)$ be a hypergraph. 
    \begin{enumerate}

    \item A \emph{merge almost-disjoint hyperedges at $s$}  operation:
         \begin{enumerate}
            \item pick a pair of hyperedges $e, f\in \delta_H(s)$ such that $e\cap f=\{s\}$,
            \item replace $e$ and $f$ with their union in $H$, i.e. $E := E - e - f + (e \cup f)$.
        \end{enumerate}
        \item A \emph{trim hyperedge at $s$} operation:
        \begin{enumerate}
            \item pick a hyperedge $e\in \delta_H(s)$,
            \item delete the vertex $s$ from $e$, i.e. $E := E - e + (e - s)$.
        \end{enumerate}
        \item A hypergraph $H'=(V, E_{H'})$ is obtained by an {h-splitting-off} operation at $s$ from $H$ if $H'$ is obtained from $H$ by either a merge almost-disjoint hyperedges at $s$ operation or a trim hyperedge at $s$ operation.

        \item A hypergraph $H^*=(V+s, E^*)$ is a \emph{complete h-splitting-off} at $s$ from $H$ if $d_{H^*}(s) = 0$ and $H^*$ is obtained from $H$ by repeatedly applying h-splitting-off operations at $s$.
    \end{enumerate}
\end{definition}

\begin{figure}[htb]
    \centering
    \includegraphics[width=0.9\linewidth]{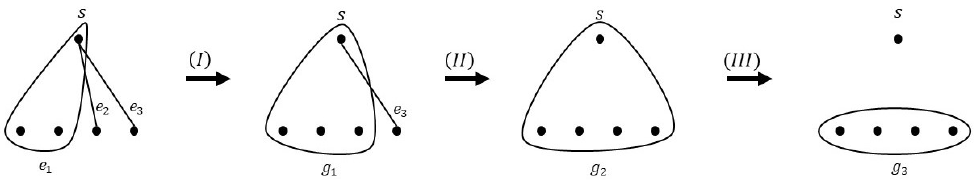}
    \caption{Example, from \cite{Bérczi_Chandrasekaran_Király_Kulkarni_2023}, of complete $h$-splitting-off at a vertex $s$ from a hypergraph. 
     Operations (I) and (II) correspond to h-merge almost-disjoint hyperedges operations and Operation (III)
       corresponds to an h-trim hyperedge operation.}
    \label{fig:hypergraph-splitting-off}
\end{figure}

The following theorem by Bérczi, Chandrasekaran, Király, and Kulkarni  shows that for every hypergraph, there exists a complete h-splitting-off at a specified vertex so that all pairwise connectivities amongst the remaining vertices are preserved. 

\begin{restatable}{theorem}{thmhsplittingoff}\label{thm:complete-splitting-off}
\cite{Bérczi_Chandrasekaran_Király_Kulkarni_2023}
    Given a hypergraph $H=(V+s, E)$, there exists a hypergraph $H^*$ that is a complete h-splitting-off at $s$ from $H$ such that $\lambda_{H^*}(u,v) = \lambda_{H}(u,v)$ for every distinct $u,v \in V$.
\end{restatable}

We emphasize that Theorem \ref{thm:complete-splitting-off} specialized for graphs does not correspond to Mader's Theorem (i.e., Theorem \ref{thm:mader}). This is because, the h-splitting-off operation applied to a graph could give only a hypergraph (and not necessarily a graph). 

\section{Connectivity Preserving Hypergraph Splitting-Off via Element-Connectivity Preserving Graph Reductions}\label{sec:alternative-proof}

Bérczi, Chandrasekaran, Király, and Kulkarni \cite{Bérczi_Chandrasekaran_Király_Kulkarni_2023} proved \Cref{thm:complete-splitting-off} as an application of a more general result of Bernáth and Király \cite{Bernath-Kiraly} on \emph{function covers} and \emph{skew-supermodularity}. A short elementary proof can also be shown by following the proof of Bern\'{a}th and Király for the specific application. 
While \cite{CK14}'s proof for element-connectivity reduction is based on simple flow and cut arguments, an alternate proof appears in \cite{CRX15} that relies on more abstract arguments. 
In this section, we give a proof of \Cref{thm:complete-splitting-off} using \Cref{thm:elem-conn-preserving-reduction}.

\paragraph{Intuition:} Before we give a formal proof we outline the core idea which we believe is simple. Let $H=(V,E)$ be a hypergraph and let $G_H=(V \cup E, E')$ be its bipartite graph representation. We view $G_H$ as an element-connectivity instance in which $V$ is the set of terminals and $E$ is the set of non-terminals. It is easy to see that $\kappa'_{(G_H, V)}(a,b) = \lambda_H(a,b)$ for all distinct $a,b \in V$. Let $s \in V$. Suppose we want to eliminate vertex $s$ from the hypergraph $H$ without affecting the hyperedge connectivity between other vertices. This is the same as eliminating terminal $s$ from $G_H$ without affecting the element-connectivity between the terminals in $V - \{s\}$. How can we eliminate the terminal vertex $s$ in the element-connectivity setting? Consider an arbitrary pair of distinct terminals $a,b \in V -\{s\}$. We can make $s$ a \emph{non-terminal} vertex but we endow it with a \emph{vertex capacity} equal to its degree $\text{deg}(s)$; recall Lemma~\ref{lem:elem-basic} on computing the element-connectivity. Since the basic element-connectivity reduction operations only work with non-terminals with unit capacity, we need to handle $s$ having capacity $\text{deg}(s)$. We employ a standard trick where we replace $s$ with a gadget graph that simulates its capacity. A natural choice is to use a clique of $d$ non-terminals but one can also use a $d \times d$ grid in planar settings (see Figure \ref{fig:elem-conn-grid}); this idea was already used in \cite{CK14} to eliminate a terminal in an application. Replacing $s$ by a gadget graph with new non-terminals results in an element-connectivity instance/graph that has edges between non-terminals. We can then use the element-connectivity reduction operations to eliminate the edges between non-terminals without disrupting all pairwise element-connectivities; eliminating edges between non-terminals gives a new graph $G'$ that can then be viewed as a hypergraph $H'$ on $V-\{s\}$; since all pairwise element-connectivities were preserved from $G_H$ to $G'$, we obtain that all pairwise edge-connectivities are preserved from $H$ to $H'$. The proof is completed by showing that the element-connectivity preserving reduction operations correspond to h-splitting-off operations in the original hypergraph. 
We now restate and prove \Cref{thm:complete-splitting-off}. 

\begin{figure}[htb]
    \centering
    \includegraphics[width=0.9\linewidth]{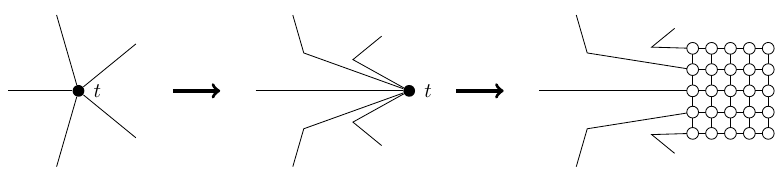}
    \caption{Replacing a terminal by a grid of non-terminals, figure from \cite{CK14}.}
    \label{fig:elem-conn-grid}
\end{figure}

\thmhsplittingoff*
\begin{proof}

Let $e_1, e_2, \ldots, e_d$ be the hyperedges incident to $s$ in $H$. 
We now describe the construction of the hypergraph $H^*$ from $H$ (see Figure \ref{fig:proof-construction} for an example). 
\begin{enumerate}
\item Let $G_0:=((V+s)\cup E, E_0)$ be the bipartite representation of $H$. We note that the nodes $e_1, e_2, \ldots, e_d$ are adjacent to $s$ in $G_0$. 
\item We obtain the graph $G_1$ from $G_0$ by replacing the vertex $s$ with a clique on a set $S_1:=\{s_1, s_2, \ldots, s_d\}$ of $d$ vertices and replacing the edges $\{e_i,s\}$ with the edges $\{e_i, s_i\}$ for every $i\in [d]$. 
\item We consider the graph $G_1$ with $V$ being the set of terminals and obtain the graph $G_2$ from $G_1$ by performing element-connectivity preserving reduction operations on all edges in $E_{G_1}[S_1]$. 
Denote the set of vertices in $G_2$ that are obtained as a result of the reduction operations as $S_2$, i.e., $S_2:=V(G_2)\setminus (V\cup E)$. 
\item We consider the graph $G_2$ with $V$ being the set of terminals and obtain the graph $G_3$ from $G_2$ by deleting as many edges incident to $S_2$ as possible while preserving element-connectivity between all terminal pairs. Since we only performed edge-deletion operations to obtain $G_3$ from $G_2$, the vertex set of $G_2$ and $G_3$ are the same. 
\item We obtain the graph $G_4$ from $G_3$ by contracting each vertex in $S_2$ with all its neighbors. We observe that the graph $G_4$ is bipartite with $V\subseteq V(G_4)$. 
\item Construct the hypergraph $H^*$ on vertex set $V$ associated with the bipartite graph $G_4$, i.e., for each node $e\in V(G_4)\setminus V$, add a hyperedge over the set of vertices that are adjacent to node $e$ in $G_4$. Add an isolated vertex $s$ to $H^*$. 
\end{enumerate}

\begin{figure}[htb]
    \centering
    \includegraphics[width=\linewidth]{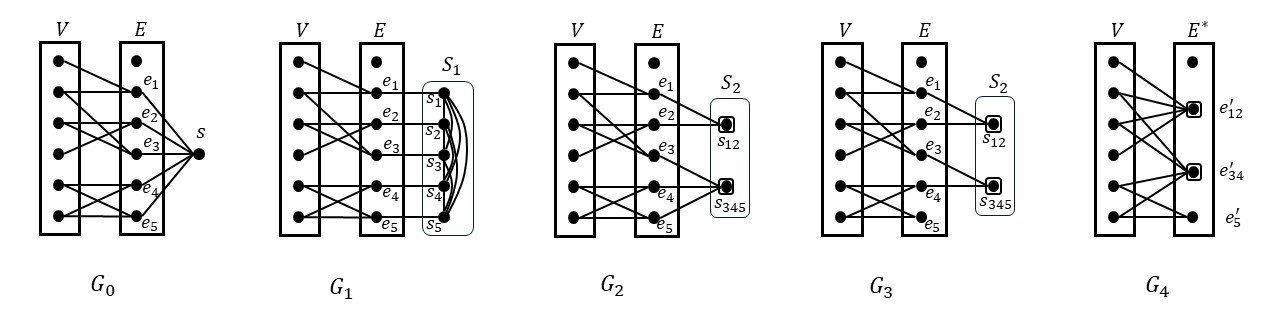}
    \caption{Example construction of the graphs in the proof of Theorem \ref{thm:complete-splitting-off}. For simplicity, we assume that the degree of $s$ in $H$ is $5$ with $e_1, e_2, \ldots, e_5$ being the hyperedges incident to $s$ in $H$. The graph $G_0$ is the bipartite representation of the hypergraph $H$. The graph $G_1$ is obtained from $G_0$ by replacing the vertex $s$ with a complete graph on a set $S_1=\{s_1, s_2, \ldots, s_5\}$ of $5$ vertices and replacing the edges $\{e_i,s\}$ with the edges $\{e_i, s_i\}$ for every $i\in [5]$. The graph $G_2$ is obtained from $G_1$ by performing element-connectivity preserving reduction operations on all edges of the clique $E_{G_1}[S_1]$. We denote the set of vertices in $G_2$ that are obtained as a result of the reduction operations as $S_2$, i.e., $S_2:=V(G_2)\setminus (V\cup E)$. 
    The graph $G_3$ is obtained from $G_2$ by performing as many element-connectivity preserving \emph{deletion} operations on the edges incident to $S_2$ as possible: in the shown example, the edge between $e_5$ and $S_2$ is deleted as a consequence of this step. The graph $G_4$ is obtained from $G_3$ by contracting each vertex in $S_2$ with all its neighbors.}
    \label{fig:proof-construction}
\end{figure}

Lemma \ref{lem:preserve-connectivity} shows that all pairwise connectivities remain the same between $H^*$ and $H$. Lemma \ref{lem:h-splitting-off} shows that the hypergraph $H^*$ is a complete h-splitting-off at $s$ from $H$. These two lemmas complete the proof of the theorem.
\end{proof}

\begin{lemma}\label{lem:preserve-connectivity}
    $\lambda_{H^*}(u,v)=\lambda_H(u,v)$ for every distinct $u, v\in V$.
\end{lemma}
\begin{proof}
    We have the following:
    \begin{enumerate}
        \item Firstly, $\kappa'_{(G_0, V+s)}(u,v)=\lambda_H(u,v)$ for every distinct $u,v\in V$. This is by Lemmas \ref{lem:elem-basic} and \ref{lem:hyperedge-connectivity}. 
        
        \item Next, we show that $\kappa'_{(G_1, V)}(u,v)=\kappa'_{(G_0, V+s)}(u,v)$ for every distinct $u,v\in V$. This follows from the results in \cite{CK14}. We include a proof for the sake of completeness. 
        Let $u, v\in V$. Suppose we have $\kappa'_{(G_1, V)}(u,v)$ element disjoint paths between $u$ and $v$ in $(G_1, V)$. Then, these paths will go through each node in $\{e_1, e_2, \ldots, e_d\}$ at most once. Each of these paths can be rerouted through $s$ to obtain $\kappa'_{(G_0, V+s)}(u,v)$ element disjoint paths between $u$ and $v$ in $(G_0, V+s)$. Hence, $\kappa'_{(G_0, V+s)}(u,v)\ge \kappa'_{(G_1, V)}(u,v)$. Next, suppose we have $\kappa'_{(G_0, V+s)}(u,v)$ element disjoint paths in $(G_0, V+s)$. Then, these paths will go through each node in $\{e_1, e_2, \ldots, e_d\}$ at most once. If a path uses the vertex $s$, then it has to be via $e_i-s-e_j$ for some $i, j\in [d]$. We can reroute such a path in $(G_1, V)$ via $e_i-s_i-s_j-e_j$. Such a rerouting of all paths leads to $\kappa'_{(G_0, V+s)}(u,v)$ element disjoint paths in $(G_1, V)$. Hence, $\kappa'_{(G_0, V+s)}(u,v)\le \kappa'_{(G_1, V)}(u,v)$. 
        
        \item We have that $\kappa'_{(G_2, V)}(u,v)=\kappa'_{(G_1,V)}(u,v)$ for every distinct $u, v\in V$. This is because, the graph $G_2$ was obtained from the graph $G_1$ via element-connectivity preserving reduction operations on the edges of $S_1$. 
        
        \item We have that $\kappa'_{(G_3, V)}(u,v)=\kappa'_{(G_2,V)}(u,v)$ for every distinct $u, v\in V$. This is because, the graph $G_3$ was obtained from the graph $G_2$ via element-connectivity preserving deletion operations on the edges incident to $S_2$. 
        
        \item Next, we show that $\kappa'_{(G_4, V)}(u,v)=\kappa'_{(G_3,V)}(u,v)$ for every distinct $u, v\in V$. 
        Let $F=\{f_1, f_2, \ldots, f_t\}$ be the set of edges incident to $S_2$ in $G_3$. 
        We recall that $G_3$ was obtained from $G_2$ by performing as many element-connectivity preserving deletion operations on the edges incident to $S_2$ as possible. Thus, none of the edges in $F$ 
        can be deleted to preserve element-connectivity between all terminal pairs. 
        
        Consider the sequence of graphs $G_0':=G_3$, $G_1':=G_0'/f_1$, $G_2':=G_1'/f_2$, $\ldots$, $G_t':=G_{t-1}'/f_t$. We observe that $G_t'=G_4$. We will prove that $\kappa_{(G_i',V)}(u,v)=\kappa_{(G_{i-1}',V)}(u,v)$ for every $u, v\in V$ and $i\in [t]$ by induction on $i$. For the base case, consider $i=1$: We recall that $G_3$ was obtained from $G_2$ by performing as many element-connectivity preserving deletion operations on the edges incident to $S_2$ as possible. Thus, deleting the edge $f_1$ does not preserve element-connectivity between all terminal pairs. Consequently, by Theorem \ref{thm:elem-conn-preserving-reduction}, contracting the edge $f_1$ preserves element-connectivty between all terminal pairs. Hence, $\kappa_{(G_1',V)}(u,v)=\kappa_{(G_{0}',V)}(u,v)$ for every $u, v\in V$. Next, we prove the induction step for $i\ge 2$. We note that $G_{i-1}'$ is a minor of $G_0$ after a sequence of element-connectivity preserving reduction operations. Consider the edge $f_i'$ corresponding to $f_i$ in $G_{i-1}'$. We claim that deleting $f_i'$ cannot preserve element-connectivity between all terminal pairs: if so, then by Lemma \ref{lem:deletions-later-allow-deletions-now}, deleting $f_i$ from $G_0=G_3$ would have preserved element-connectivity between all terminal pairs, a contradiction to the fact that none of the edges in $F$ can be deleted to preserve element-connectivity between all terminal pairs. Hence, by Theorem \ref{thm:elem-conn-preserving-reduction}, contracting the edge $f_i'$ in the graph $G_{i-1}'$ preserves element-connectivity between all terminal pairs. Hence, $\kappa_{(G_i',V)}(u,v)=\kappa_{(G_{i-1}',V)}(u,v)$ for every $u, v\in V$. 
            
        
        \item Finally, we have that $\lambda_{H^*}(u,v)=\kappa'_{(G_4, V)}(u,v)$ for every distinct $u, v\in V$. This is by Lemmas \ref{lem:elem-basic} and \ref{lem:hyperedge-connectivity}. 
    \end{enumerate}
\end{proof}

\begin{lemma}\label{lem:h-splitting-off}
    The hypergraph $H^*$ is a complete h-splitting-off at $s$ from $H$. 
\end{lemma}
\begin{proof}
    We observe that $d_{H^*}(s)=0$ by construction. 
    We need to show that $H^*$ is obtained from $H$ by a sequence of h-splitting-off operations at $s$. 

    We observe that each node $e_1, e_2,\ldots, e_d$ is adjacent to a unique node of $S_1$ in $G_1$. Moreover, each node $e_1, e_2, \ldots, e_d$ is adjacent to a unique node of $S_2$ in $G_2$. We define the following: 
    \begin{enumerate}
    \item For each node $a\in S_2$, let 
    $F_{a}:=\{e_i: i\in [d], e_i\text{ is adjacent to $a$ in } G_3\}$. 
    \item Let $F_0:=\{e_i: i\in [d], e_i\text{ is not adjacent to any node of $S_2$ in } G_3\}$. 
    \end{enumerate}
    We consider the following procedure: 
    \begin{enumerate}
        \item Start with the current hypergraph being $H$.
        \item Pick an arbitrary ordering of the hyperedges in $F_0$, say $e_{i_1}, e_{i_2}, \ldots, e_{i_{t_0}}$. 
        \item For $j=1, 2, \ldots, t_0$:
        \begin{enumerate}
            \item Update the current hypergraph by trimming $s$ from $e_{i_j}$.
        \end{enumerate}
        \item Pick an arbitrary ordering of the vertices in $S_2$, say $a_1, a_2, \ldots, a_r$. 
        \item For each $i=1, 2, \ldots, r$:
        \begin{enumerate}
            \item Pick an arbitrary ordering of the hyperedges in $F_{a_i}$, say $e_{j_1}, e_{j_2}, \ldots, e_{j_{t_i}}$. 
            \item For each $k=2, \ldots, t_i$: 
            \begin{enumerate}
                \item Update the current hypergraph by merging $e_{j_k}$ with $e_{j_{k-1}}$ and label the merged hyperedge as $e_{j_{k-1}}$. 
            \end{enumerate}
            \item Update the current hypergraph by trimming $s$ from $e_{j_1}$. 
        \end{enumerate}
    \end{enumerate}
    The hypergraph obtained at the end of this procedure is exactly $H^*$. Each operation to update the current hypergraph in the above procedure is either a trim operation or a merge operation. It remains to show that the merge operation in the above procedure picks only almost-disjoint hyperedges. Claim \ref{claim:near-disjoint} proves this. 
\end{proof}

    \begin{claim}\label{claim:near-disjoint}
        Let $a\in S_2$. Let $e_i$ and $e_j$ be adjacent to $a$ in $G_3$. Then, the hyperedges $e_i$ and $e_j$ are almost-disjoint in $H$, i.e., $s$ is the only common vertex in the hyperedges $e_i$ and $e_j$. 
    \end{claim}
    \begin{proof}
        We recall that the graph $G_3$ was obtained from $G_2$ by performing as many element-connectivity preserving deletion operations on the edges incident to $S_2$ as possible. Thus, none of the edge incident to $S_2$ in $G_3$ can be deleted to preserve element-connectivity between all terminal pairs. In particular, by Theorem \ref{thm:elem-conn-preserving-reduction}, the edge $\{e_i,a\}$ can be contracted to preserve element-connectivity between all terminal pairs. Let $G_3'$ be the graph obtained from $G_3$ by contracting the edge $\{e_i, a\}$ with $a'$ being the contracted vertex. Let $G_3''$ be the graph obtained from $G'$ by deleting the edge $\{e_j, a'\}$ (such an edge exists in $G'$ since $\{e_j, a\}$ exists in $G_3$). 
        
        For the sake of contradiction, suppose there exists a vertex $w\in (e_i\cap e_j)\setminus \{s\}$. We will exploit such a vertex $w$ to show that $G_3''$ preserves element connectivity between all terminal pairs. 
        Let $u, v \in V$ be arbitrary terminals. Let $\kappa':=\kappa'_{(G'_3, T)}(u,v)$. 
        We need to show that $\kappa'_{(G''_3, T)}(u,v)=\kappa'$. 
        Since $G''_3$ is obtained from $G'_3$ by deleting an edge, we have that $\kappa'_{(G''_3, T)}(u,v)\le\kappa'$. We need to show that deleting the edge $\{e_j,a'\}$ from $G'_3$ does not reduce the maximum number of element-disjoint paths between $u$ and $v$. 
        Let $P_1, P_2, \ldots, P_{\kappa'}$ be the element-disjoint paths between $u$ and $v$ in $(G'_3, T)$. If none of these paths use the edge $\{e_j,a'\}$, then these are also element-disjoint paths in $(G''_3,T)$ and we are done. Hence, some of these paths use the edge $\{e_j, a'\}$. Since the nodes $e_j, a'$ and the edge $\{e_j, a'\}$ are elements in $(G'_3, T)$, exactly one path among $P_1, P_2, \ldots, P_{\kappa'}$ uses the edge $\{e_j, a'\}$. Without loss of generality, let $P_1$ use the edge $\{e_j, a'\}$. Without loss of generality, let $P_1$ be $u_1=u, u_2, \ldots, u_k, u_{k+1}=e_j, u_{k+2}=a', u_{k+3}, \ldots, u_{|P_1|}=v$.  We now construct an alternate path $P_1'$ between $u$ and $v$ in $G'_3$ that does not use the edge $\{e_j, a'\}$ such that $P_1', P_2, \ldots, P_{\kappa'}$ are element-disjoint. Consider the vertex $w\in (e_i\cap e_j)\setminus \{s\}$ and the corresponding node $w\in V(G'_3)$. By construction, the vertex $w$ is a terminal vertex in $(G', T)$ and moreover, the edges $\{e_j, w\}$ and $\{e_i, w\}$ are present in $G_3$. Consequently, the edges $\{e_j,w\}$ and $\{a', w\}$ are present in $G'_3$. Thus, we have the path $P_1'$ given by $u_1=u, u_2, \ldots, u_k, u_{k+1}=e_j, w, u_{k+2}=a', u_{k+3}, \ldots, u_{|P_1|}=v$ in $G'_3$. We note that $P_1'$ does not use the edge $\{e_j, a'\}$. It remains to show that $P_1', P_2, P_3, \ldots, P_{\kappa'}$ are element-disjoint paths. In particular, we need to show that the edges $\{e_j, w\}$ and $\{w, a'\}$ and the vertices $e_j$ and $a'$ are not present in any of the other paths $P_2, P_3, \ldots, P_{\kappa'}$. However, since the vertices $e_j$ and $a'$ are already present in $P_1$, they are not present in $P_2, P_3, \ldots, P_{\kappa'}$. 
        
        The fact that the edge $\{e_j, a'\}$ can be deleted from $G'_3$ to preserve element-connectivity between all terminal pairs leads to a contradiction as follows: 
        consider the graph $G_3$ with terminal set $T$. The graph $G'_3$ is a minor of $G_3$ after an element-connectivity preserving reduction operation. We have shown that deleting the edge $\{e_j, a'\}$ from $G_3'$ preserves element-connectivity between all terminal pairs. By Lemma \ref{lem:deletions-later-allow-deletions-now}, it follows that deleting $\{e_j,a\}$ from $G_3$ preserves element-connectivity between all terminal pairs. This contradicts the fact that none of the edges incident to $S_2$ in $G_3$ can be deleted to preserve element-connectivity.

    \end{proof}


\paragraph{Acknowledgments.}
 Karthekeyan Chandrasekaran and Shubhang Kulkarni would like to thank Krist\'{o}f B\'{e}rczi and Tam\'{a}s Kir\'{a}ly for  discussions on the connections between hypergaph splitting-off and element connectivity.
\bibliographystyle{alpha}
\bibliography{references}


\end{document}